\documentclass[12pt]{article}
\usepackage{slashbox}
\usepackage{multirow}
\usepackage[T1]{fontenc}
\usepackage{algpseudocode}
\usepackage{algorithm}

\usepackage{graphicx}
\usepackage{enumerate}
\usepackage{rotating}
\usepackage{tikz}
\usepackage{amsmath, amsthm, amsfonts, amssymb}

\newtheorem{theorem}{Theorem}
\newtheorem{lemma}{Lemma}

\usepackage{color}

\begin{document}
\title{\bf Scheduling of unit-length jobs with cubic incompatibility graphs on three uniform machines\footnote{This project has been 
partially supported by Narodowe Centrum Nauki under contract 
DEC-2011/02/A/ST6/00201}}
\author{Hanna Furma\'nczyk\footnote{Institute of Informatics,\ University of Gda\'nsk,\ Wita Stwosza 57, \ 80-952 Gda\'nsk, \ Poland. \ e-mail: hanna@inf.ug.edu.pl},  \ Marek 
Kubale\footnote{Faculty of Electronics, Telecommunications and Informatics,\ Technical University of Gda\'nsk,\ Narutowicza 11/12, \ 80-233 Gda\'nsk, \ Poland. \ e-mail: kubale@eti.pg.gda.pl}}
\date{}

\maketitle
\begin{abstract}
In the paper we consider the problem of scheduling $n$ identical jobs on 3 uniform machines with speeds $s_1, s_2,$ and $s_3$ to 
minimize the schedule length. We assume that jobs are subjected to some kind of mutual exclusion constraints, modeled by a cubic incompatibility graph. 
We show that if the graph is 2-chromatic then the problem can be solved in $O(n^2)$ time. If the graph 
is 3-chromatic, the problem becomes NP-hard even if $s_1>s_2=s_3$. However, in this case there exists a $4/3$-approximation algorithm running in $O(n^3)$
time. Moreover, this algorithm 
solves the problem almost surely to optimality if $3s_1/4 \leq s_2 = s_3$.
\end{abstract}
\section{Introduction}
Imagine you have to arrange a dinner for, say 30, people and you have at your disposal 3 round tables with different numbers of seats 
(not greater than 15). 
You know that each of your guests is in bad relations with exactly 3 other people. Your task is to assign the people to the tables in such a way that no 
two of them being in bad relations seat at the same table. In the paper we show how to solve this and related problems.

Our problem can be expressed as the following scheduling problem. Suppose we have $n$ identical jobs $j_1,\ldots, j_n$, so we assume that they all 
have unit execution times, in symbols $p_i=1$, to be processed on three non-identical machines $M_1, M_2,$ and $M_3$. These machines run at different speeds $s_1, s_2,$ and $s_3$, respectively. 
However, they are \emph{uniform} in the sense that if a job is executed on machine $M_i$, it takes $1/s_i$ time units to be completed. It refers to the 
situation where the machines are of different generations, e.g. old and slow, new and fast, etc.

Our scheduling model would be trivial if all the jobs were compatible. Therefore we assume that some pairs of jobs cannot be processed on the same machine 
due to some technological constraints. More precisely, we assume that each job is in conflict with exactly three other jobs. Thus the underlying \emph{incompatibility graph} $G$ whose 
vertices are jobs and edges correspond to pairs of jobs being in conflict is cubic. For example, all figures in this paper comprise cubic graphs. By the handshaking 
lemma, the number of 
jobs $n$ must be even. A load $L$ on machine $M_i$ requires the processing time $P_i(L)= |L|/s_i$, and all jobs are ready for 
processing at the same time. 
By definition, each load forms an independent set (color) in $G$. Therefore, in what follows we will be using the terms 
job/vertex and color/independent set interchangeably. Since all tasks have to be executed, the problem is to find a 3-coloring, i.e. a decomposition of $G$ into 3 
independent sets 
$I_1, I_2,$ and $I_3$ such that the schedule length $C_{\max} = \max\{P_i(I_i): i=1,2,3\}$ is minimized, in symbols $Q3|p_i=1, G=cubic|C_{\max}$. 

In this paper we assume three machines for the following reason. If there is only one machine then there is no solution. If there are two machines, the problem becomes trivial 
because it is solvable only if $G$ is bipartite and it has only one solution since there is just one decomposition of $G$ into sets $I_1$ and $I_2$, each 
of size $n/2$. If, however, there are three machines and $G$ is 3-chromatic, our problem becomes NP-hard. Again, if $G$ is 4-chromatic, there is no solution.

There are several papers devoted to scheduling in the presence of mutual exclusion constraints. Boudhar in \cite{boudhar1,boudhar2} studied the problem 
of batch scheduling
with complements of bipartite and split 
graphs, respectively. Finke et al. \cite{fjqs} considered the 
problem with complements of interval graphs. Our problem can also be viewed as  a particular variant of scheduling with conflicts \cite{conf}.
In all the papers the authors assumed identical parallel machines. However, to the best of our knowledge little work has 
been done on scheduling problems with uniform machines involved (cf. Li and Zhang \cite{li_zhang}).  

The rest of this paper is split into two parts depending on the chromaticity of cubic graphs. In Section \ref{sec2} we consider 2-chromatic graphs. In particular, we give an $O(n^2)$-time algorithm for optimal scheduling of such graphs. Section \ref{sec3} is devoted to 3-chromatic graphs. In particular, we give an NP-hardness proof and an approximation algorithm with good performance guarantee. Our algorithm runs in $O(n^3)$ time to produce a solution of value less than 4/3 times optimal, provided that $s_1> s_2= s_3$.
Moreover, this algorithm solves the problem almost surely to optimality if $3s_1/4=s_2=s_3$. Finally, we discuss possible extensions of our model to disconnected graphs.

\section{Scheduling of 2-chromatic graphs}\label{sec2}

We begin with introducing some basic notions concerning graph coloring. A graph $G = (V,E)$ is said to be \emph{equitably $k$-colorable} if and 
only if its vertex set can be partitioned into independent sets $V_1, \ldots, V_k \subset V$ such that $||V_i| - |V_j|| \leq 1$ for all 
$i, j = 1, \ldots, k$. The smallest $k$ for which $G$ admits such a coloring is called the \emph{equitable chromatic number} of $G$ and 
denoted $\chi_=(G)$. Graph $G$ has a \emph{semi-equitable coloring}, if there exists a partition of its vertices into independent sets 
$V_1,\ldots, V_k \subset V$ such that one of these subsets, say $V_i$, is of size $\notin \{\lfloor n/k\rfloor, \lceil n/k \rceil\}$, and the remaining subgraph $G-V_i$ is equitably $(k-1)$-colorable.

Let us recall some basic facts concerning colorability of cubic graphs. It is well known from Brooks theorem \cite{brooks}
that for any cubic graph $G \neq K_4$ we have $\chi(G) \leq 3$. On the other hand, Chen et. al. \cite{clw} proved that every 3-chromatic cubic 
graph can be equitably colored without introducing a new color. Moreover, since a connected cubic graph $G$ with $\chi(G) = 2$ is a bipartite 
graph with partition sets of equal size, we have the equivalence of the classical and equitable chromatic numbers for 2-chromatic cubic graphs. 
Since the only cubic graph for which the chromatic number is equal to 4 is the complete graph $K_4$, we have

\begin{equation}
2 \leq \chi_=(G) = \chi(G) \leq 4	\label{eq1}
\end{equation}
for any cubic graph. Moreover, from (\ref{eq1}) it follows that for any cubic graph $G \neq K_4$, we have

\begin{equation}
n/3 \leq \alpha(G) \leq n/2
\end{equation}
where $\alpha(G)$ is the independence number of $G$. Note that the upper bound is tight only if $G$ is bipartite.

Let $\mathcal{Q}_k$ denote the class of connected $k$-chromatic cubic graphs and let $\mathcal{Q}_k(n) \subset \mathcal{Q}_k$ stand for the 
subclass of cubic graphs on $n$ vertices, $k = 2,3,4$. Clearly, $\mathcal{Q}_4 = \{K_4\}$. In what follows we will call the graphs belonging to 
$\mathcal{Q}_2$ \emph{bicubic}, and the graphs belonging to $\mathcal{Q}_3$ - \emph{tricubic}. 

As mentioned, if $G$ is bicubic then any 2-coloring of it is equitable and there may be no equitable 3-coloring (cf. $K_{3,3}$). On the other 
hand, all graphs in $\mathcal{Q}_2(n)$ have a semi-equitable 3-coloring of type $(n/2, \lceil n/4 \rceil,\lfloor n/4 \rfloor)$. Moreover, they 
are easy colorable in linear time while traversing in a depth-first search (\texttt{DFS}) manner. 

Let $s_i$ be the speed of machine $M_i$ for $i = 1,2,3$, and let $s = s_1+s_2+s_3$. Without loss of generality we assume that $s_1 \geq s_2 \geq s_3$. 
If there are just 6 jobs to schedule then the incompatibility graph $G = K_{3,3}$ and there is only one decomposition of it into 3 independent 
sets shown in Fig.~\ref{rys1}(a), and there is only one decomposition of $G$ into 2 independent sets shown in Fig.~\ref{rys1}(b), of course up to isomorphism. The length of minimal 
schedule is $\min\{\max\{3/s_1, 2/s_2, 1/s_3 \}, 3/s_2\}$. Therefore, we assume that our graphs have at least 8 vertices.

Notice that if $s_1 \geq s_2 + s_3$ then as many as possible jobs should be placed on $M_1$. The maximal number of jobs on the first machine is $n_1 = n/2$. The remaining $n/2$ jobs should be assigned to $M_2$ and $M_3$ in quantities proportional to their speeds, more precisely 
in quantities $n_2 =\lceil .5ns_2/(s_2+ s_3)\rceil$ and $n_3 = n/2- n_2$, respectively. If $s_1 < s_2+ s_3$ then the number of jobs on machine 
$M_i$ should be proportional to its speed $s_i$. In such an ideal case the total processing times of all the loads would be the same. However, the 
numbers of jobs on machines must be integer. Therefore, we must check which of the three variants of a schedule, i.e. with round-up and/or round-down 
on $M_1$ and $M_2$, guarantees a better solution. This leads to the following algorithm for optimal scheduling of bicubic graphs.

\begin{figure}[htb]
\begin{center}
\includegraphics[scale=1]{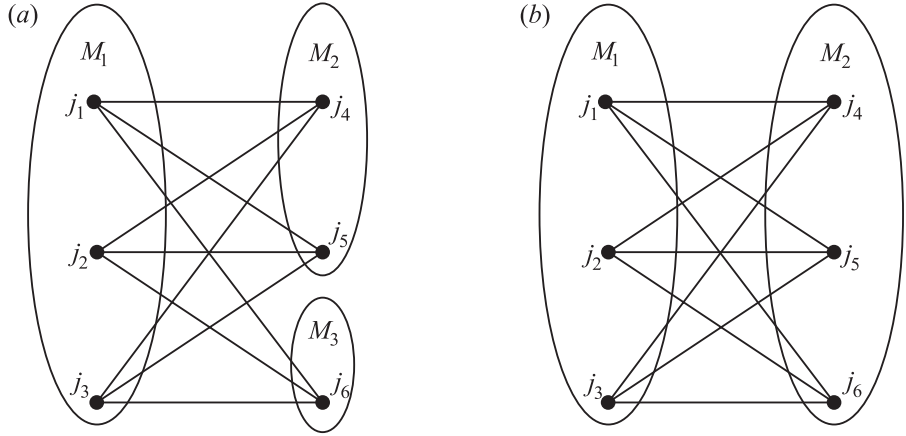} 
%\vspace{-4cm}
\caption{Two decompositions of $K_{3,3}$: (a) into 3 independent sets, (b) into 2 independent sets.}
\label{rys1}
\end{center}
\end{figure}

\begin{algorithm}
\caption{Scheduling of bicubic graphs}
\begin{algorithmic}
\Require {Graph $G \in \mathcal{Q}_2(n)$, $G \neq K_{3,3}$ and machine speeds $s_1, s_2, s_3$ such that $s_1 \geq s_2 \geq s_3$.}
\Ensure {Optimal schedule.}
\begin{enumerate}
\item If $s_1 < s_2+ s_3$ then go to Step 5.
\item Find an $(I, J)$-coloring of graph $G$.
\item Split color $J$ into 2 subsets: $B$ of size $n_2 =\lceil .5ns_2/(s_2+ s_3)\rceil$ and $C$ of size $n_3 = n/2- n_2$.
\item Assign $M_1 \leftarrow I$, $M_2 \leftarrow B$, $M_3 \leftarrow C$ and stop.
\item Calculate approximate numbers of jobs $(n_1, n_2, n_3)$ to be processed on $M_1, M_2, M_3$ in an ideal schedule, as follows:
$$n_1 = ns_1/s, n_2 = ns_2/s, n_3 = ns_3/s, \text{ where } s = s_1+ s_2+ s_3.$$
\item Verify which of the following types of colorings:
$$(\lfloor n_1\rfloor,\lceil n_2 \rceil, n-\lfloor n_1\rfloor- \lceil n_2 \rceil), \ 
(\lceil n_1 \rceil, \lfloor n_2\rfloor, n-\lceil n_1 \rceil -\lfloor n_2\rfloor) \text{ or } (\lceil n_1 \rceil, \lceil n_2\rceil, n-\lceil n_1 \rceil -\lceil n_2\rceil)$$
guarantees a better solution and call it OPT.
\item Let $(A, B, C)$ be a coloring of $G$ realizing OPT obtained by using a modified \texttt{CLW} method described in \texttt{Procedure 1}.
\item Assign $M_1 \leftarrow A$, $M_2 \leftarrow B$, $M_3 \leftarrow C$.
\end{enumerate}
\end{algorithmic}
\end{algorithm}

A crucial point of \texttt{Algorithm 1} is Step 7 where we use a modified procedure due to Chen et al. \cite{clw}, which we call a \texttt{CLW} procedure. 
This procedure was used by them to prove that every tricubic graph can be equitably colored without introducing a new color.  
\texttt{CLW} relies on successive decreasing the \emph{width} of coloring, i.e. the difference between the cardinality of the largest and smallest independent 
set, one by one until a coloring is equitable. Actually, their procedure works for every 3-coloring of any bicubic graph, except for $K_{3,3}$. 
More precisely, in Step 7 of \texttt{Algorithm 1}, where we want to receive an $(A,B,C)$-coloring (named as OPT) with cardinalities of color classes $|A| \geq |B| \geq |C|$,  we have 
to start with 2-coloring of bicubic graph $G$: $(I,J)$-coloring. Next we split the color class $J$ into two: $B$ of cardinality $|B|$ and $C'$ 
of size $n/2-|B|$. Hence, we initially have $(A', B, C')$-coloring with $|A'|=n/2$, $C'=n/2-|B|$, where the largest class is clearly $A'$, 
while the smallest class is $C'$. If this coloring with the width of $|B|$ is not the desirable $(A,B,C)$-coloring with the width of $|A|-|C|$, 
then we use \texttt{CLW} for decreasing the width from $|B|$ to $|A-C|$. Let us notice that such a width decreasing step is applied only to the first and 
the third color class, without changing the cardinality of the second class which is still equal to $|B|$. The whole modified \texttt{CLW} procedure 
is given below as \texttt{Procedure 1}. The complexity of modified \texttt{CLW} is the same as the complexity of the original \texttt{CLW} procedure for making 
any 3-coloring of tricubic graph equitable, namely $O(n^2)$. This is so because the part of the algorithm responsible for decreasing the width 
of coloring by one may be done in linear time. In the nutshell, we first check if there is a pair of vertices one from the largest and the 
other from the smallest class whose colors can be simply swapped. If there is no such pair, we have to consider such a bipartite subgraph 
that swapping the vertices between its partition sets (possibly with another subset being involved in the swapping) results in decreasing the 
width of coloring. Since this step must be repeated at most $n/6$ times, the complexity of modified \texttt{CLW} procedure follows. This 
complexity dominates the running time of  \texttt{Algorithm 1}.

\floatname{algorithm}{Procedure}
\setcounter{algorithm}{0}
\begin{algorithm}
\caption{Modified \texttt{CLW} algorithm}
\begin{algorithmic}
\Require {Graph $G\in \mathcal{Q}_2(n)$, $G \neq K_{3,3}$ and integers $a \geq b \geq c$ such that $a+b+c = n$.}
\Ensure {$(A,B,C)$-coloring of $G$ such that $|A|=a, |B|=b, |C|=c$.}
\begin{enumerate}
\item Find an $(I, J)$-coloring of graph $G$.
\item Split $J$ into 2 subsets: $B$ of size $b$ and $C’$ of size $n/2-b$.
\item While $|C’| < c$ do 

decrease the width of coloring by one using the \texttt{CLW} method \cite{clw}.
\end{enumerate}
\end{algorithmic}
\end{algorithm}

The above considerations lead us to the following
\begin{theorem}
\emph{\texttt{Algorithm 1}} runs in $O(n^2)$ time to produce an optimal schedule. \hfill $\Box$
\end{theorem}

\section{Scheduling of 3-chromatic graphs}\label{sec3}

First of all notice that if $s_1= s_2= s_3$ then the scheduling problem becomes trivial since any equitable coloring of $G$ solves the problem 
to optimality. Therefore we assume that only two possible speeds are allowed for machines to run, more precisely that $s_1> s_2=s_3$. As 
previously, if there are just 6 jobs to schedule then the incompatibility graph $G = P$, where $P$ is the prism shown in Fig.~\ref{rys2}. 
There is only one decomposition of $P$ into 3 independent sets and the length of minimal schedule is $2/s_2$. Therefore, we assume that our 
graphs have at least 8 vertices.

\begin{figure}[htb]
\begin{center}
\includegraphics[scale=0.8]{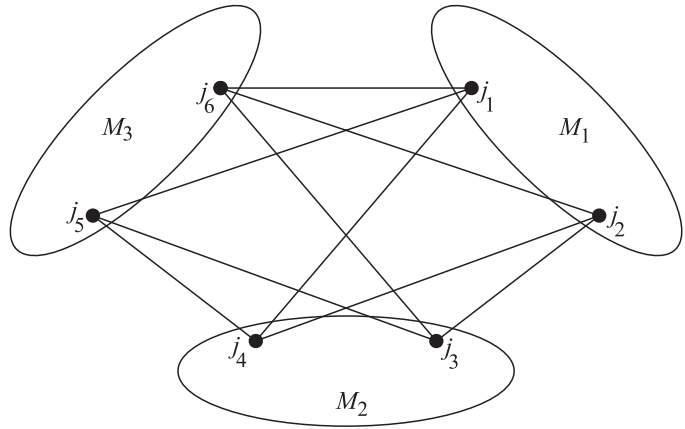} 
%\vspace{4cm}
\caption{The prism $P$ and its decompositions into 3 independent sets.}
\label{rys2}
\end{center}
\end{figure}

In the following we take advantage of the following
\begin{lemma}[Furma\'nczyk, Kubale \cite{fk}]
Let $G\in \mathcal{Q}_3(n)$ and let $k=n/10$, where $10|n$. The problem of deciding whether $G$ has a semi-equitable coloring of type 
$(4k,3k,3k)$ is \emph{NP}-complete.   \hfill $\Box$\label{lm1}
\end{lemma}

Now we are ready to prove
\begin{theorem}
The $Q3|p_i=1,G \in \mathcal{Q}_3(n)|C_{\max}$ problem is NP-hard even if $s_1> s_2= s_3$.
\end{theorem}
\begin{proof}
In the proof we will use a reduction of the coloring problem from Lemma \ref{lm1} to our scheduling problem.

So suppose that we have a tricubic graph $G$ on $n=10k$ vertices and we want 
to know whether there exists a $(4k,3k,3k)$-coloring of $G$. Given such an instance we construct the following instance for a scheduling 
decision problem: machine speeds for $M_1, M_2,$ and $M_3$ are $s_1 = 4/3, s_2= s_3= 1$ and the limit on schedule length is $3k$. The question is 
whether there is a schedule of length at most $3k$? The membership of this problem in class NP is obvious.
            
If there is a schedule of length $\leq 3k$ then it is of length exactly $3k$ since it cannot be shorter. Such a schedule implies the existence of a 
semi-equitable coloring of $G$ of type 
$(4k,3k,3k)$. 

If $G$ has a coloring of type $(4k,3k,3k)$ then our scheduling problem has clearly a solution of length $3k$. 

The NP-hardness of $Q3|p_i=1,G \in \mathcal{Q}_3(n)|C_{\max}$ follows from the fact that its decision version is NP-complete. 
\end{proof}

Since our scheduling problem is NP-hard, we have to propose an approximation algorithm for it.

\floatname{algorithm}{Algorithm}
\setcounter{algorithm}{1}
\begin{algorithm}
\caption{Scheduling of tricubic graphs}
\begin{algorithmic}
\Require {Graph $G \in \mathcal{Q}_3(n)$, $G \neq P$ and machine speeds $s_1, s_2, s_3$ such that $s_1 > s_2 = s_3$.}
\Ensure {Suboptimal schedule.}
\begin{enumerate}
\item Apply procedure \texttt{Greedy} (described in \texttt{Procedure 2}) to find an independent set $I$ of $G$. 
If $|I| < 0.4n$ then go to Step 5.
%\item Find an $(I, J)$-coloring of graph $G$.
\item If $G-I$ is not bipartite then apply procedure \texttt{FKR} (cf. \cite{fkr}) to get an independent set $A$, $|A|=|I|$, which bipartizes $G$ and 
put $I = A$.
\item Find an equitable 2-coloring $(B, C)$ of $G-I$.
\item If $s_1 \geq 2s_2$ then assign $M_1 \leftarrow I$, $M_2 \leftarrow B$, $M_3 \leftarrow C$ and stop else go to Step 6. 
\item Find any 3-coloring of $G$ (cf.\cite{skul}) and apply procedure \texttt{CLW} in order to obtain an 	equitable coloring $(A, B, C)$ of $G$. 
Go to Step 9.
\item Calculate approximate numbers of jobs  $(n_1, n_2, n_3)$ to be processed on $M_1, M_2, M_3$ in an ideal schedule, as follows:
$$n_1 = ns_1/s, n_2 =  n_3 = ns_2/s, \text{ where } s = s_1+ s_2+ s_3.$$
\item Verify which of the following types of colorings guarantees a better legal solution:
$$(\lfloor n_1\rfloor,\lceil n_2 \rceil, n-\lfloor n_1\rfloor- \lceil n_2 \rceil) \text{ or }
(\lceil n_1 \rceil, \lfloor n_2\rfloor, n-\lceil n_1 \rceil -\lfloor n_2\rfloor) $$

If this is the first type then let $n^* = \lfloor n_1\rfloor$ else $n^*= \lceil n_1 \rceil$.
\item If $n^* < |I|$  then  apply a modified  \texttt{CLW}  procedure to obtain a semi-equitable  coloring $(A, B, C)$ of $G$, where $|A| = n^*$.
\item Assign $M_1 \leftarrow A$, $M_2 \leftarrow B$, $M_3 \leftarrow C$.
\end{enumerate}
\end{algorithmic}
\end{algorithm}

Procedure \texttt{Greedy} repeatedly chooses a vertex $v$ of minimum degree, adds it to its current independent set and then deletes $v$ and all its 
neighbors. Its complexity is linear. The following \texttt{Procedure 2} gives a more formal description of it.
\floatname{algorithm}{Procedure}
\setcounter{algorithm}{1}
\begin{algorithm}
\caption{\texttt{Greedy}}
\begin{algorithmic}
\Require {Graph $G\in \mathcal{Q}_3(n)$.}
\Ensure {Independent set $I$ of $G$.}
\begin{enumerate}
\item Set $I = \emptyset$.
\item While $V(G) \neq \emptyset$ do
		
set $G = G-N[v]$ and $I = I \cup \{v\}$, where $v$ is a minimum degree vertex in $G$ and	$N[v]$ is its closed neighborhood.
\end{enumerate}
\end{algorithmic}
\end{algorithm}

Note that \texttt{Greedy} does not guarantee that $G-I$ is bipartite. It may happen that there remain some odd cycles in the subgraph, even if a 
big independent set is found. An example of such situation is given in Fig.~\ref{rys3}. Nevertheless, the authors proved in \cite{fkr} that 
given a graph $G \in \mathcal{Q}_3(n)$ with $\alpha(G) \geq 0.4n$, there exists an independent set $I$ of size $k$ in $G$ such that $G-I$ is 
bipartite for $\lfloor(n- \alpha(G))/2 \rfloor \leq k \leq \alpha(G)$.

\begin{figure}[htb]
\begin{center}
\includegraphics[scale=0.5]{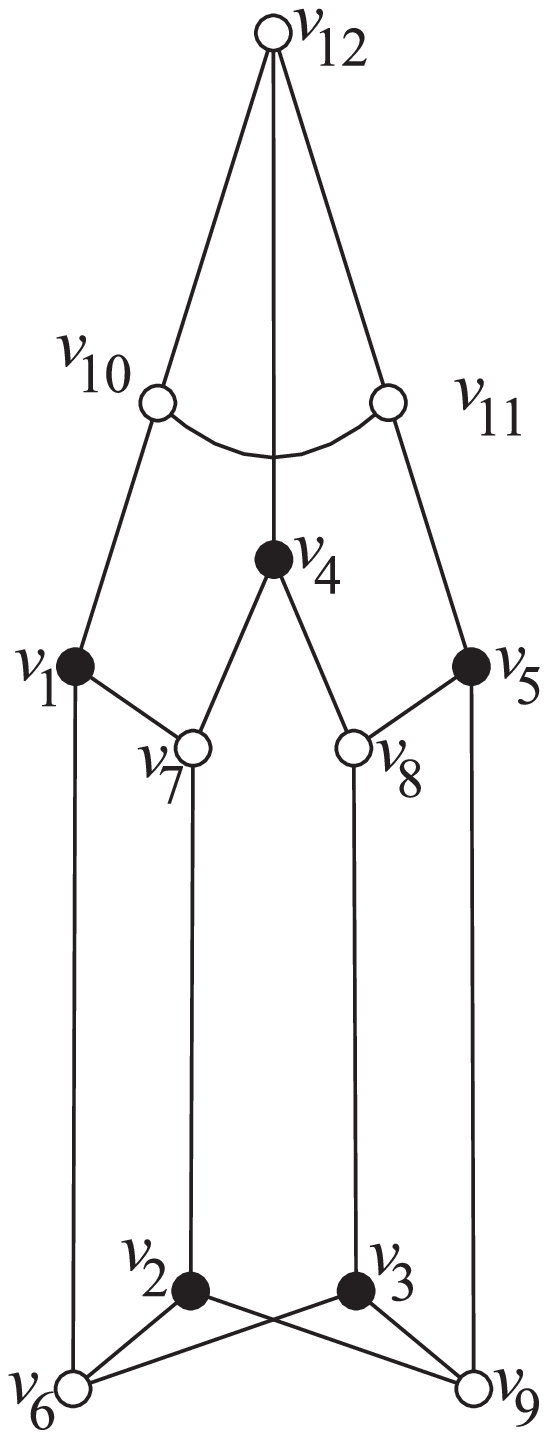} 
%\vspace{4cm}
\caption{Graph $G$ for which the \texttt{Greedy} procedure (with ties broken by choosing the vertex with smallest index) finds an independent set $I$ 
(vertices in black) such that $G-I$ contains $K_3$.}
\label{rys3}
\end{center}
\end{figure}

Now we have to prove that if independent set $|I| \geq 0.4n$ and $G-I$ is bipartite then $G-I$ is equitably 2-colorable. Indeed, assume that $|I|=0.4n$. Notice that 
$0.6n$ vertices of $G-I$ induce binary trees (some of them may be trivial) and/or graphs whose 2-core is equibipartite (even cycle possibly 
with chords). Note that deleting an independent set $I$ of cardinality $0.4n$ from a cubic graph $G$ means also that we remove 
$1.2n$ edges from the set of all $1.5n$ edges of $G$. The resulting graph $G-I$ has $0.6n$ vertices and $0.3n$ edges. Let $d_i$, $0 \leq i \leq 3$, be 
the number of 
vertices 
in $G-I$ of degree $i$. Certainly, $d_0+\ldots+d_3=0.6n$. Since the number of edges is half of the number of vertices, the number of isolated vertices, 
$d_0$, is 
equal to $d_2 + 2d_3$. If $d_0=0$, then $G-I$ is a perfect matching and its equitable coloring is obvious. 
Suppose that $d_0 > 0$. Let $P$ denote the set of isolated vertices in $G-I$. Let us consider subgraph $G-I-P$. Each vertex of degree 3 causes the difference between cardinalities 
of color classes $\leq 2$, similarly each vertex of degree 2 causes the difference at most 1. The difference between the cardinalities of color classes in any 
coloring fulfilling these conditions does not exceed  $d_2+2d_3$ in $G-I-P$. Thus, the appropriate assignment of 
colors to isolated vertices in $P$ makes the whole graph $G-I$ equitably 2-colored. 
Therefore, an equitable coloring of $G-I$ required in Step 3 of \texttt{Algorithm 1} can be obtained as follows. First we color non-isolated vertices greedily by using for example a \texttt{DFS} 
method. In the second phase we color 
isolated 
vertices with this color that has been used fewer times in the first phase. This can be accomplished in $O(n)$ time.

However, the most time consuming is Step 2, where the \texttt{FKR} procedure is invoked. This procedure is too complicated to be described here. The general idea is as follows: given 
$G\in \mathcal{Q}_3(n)$ and an independent set $I$ of size at least $0.4n$ such that $G-I$ is 3-chromatic, we transform it step by step into an independent set $I'$ such that $|I'| = |I|$ 
and $G-I'$ is 2-chromatic (see \cite{fkr} for details). Since one step of swapping two vertices between $I$ and $V-I$ requires $O(n^2)$ time, the complexity of  \texttt{FKR} is $O(n^3)$.

The above considerations lead us to the following
\begin{theorem}
\emph{\texttt{Algorithm 2}} runs in $O(n^3)$ time. \hfill $\Box$  
\end{theorem}

Now we shall prove two fact concerning the performance guarantees for \texttt{Algorithm 2}.
\begin{theorem}
\emph{\texttt{Algorithm 2}} returns a solution of value less than  $\frac{4}{3}C^*_{\max}$.
\end{theorem}
\begin{proof}
Let Alg$_2(G)$ be the length of a schedule produced by \texttt{Algorithm 2} when applied to incompatibility graph $G$, and let $C^*_{\max}(G)$ be the length of an optimal schedule. 

If $s_1 \geq 2s_2$ then it is natural to load as many jobs as possible on the fastest batch machine $M_1$. By inequality (2) the maximal possible number of jobs on $M_1$ is less than $n/2$. 
Therefore, the schedule length on $M_1$ is less than $\frac{1}{2}n/s_1 \leq \frac{1}{4}n/s_2$. In an optimal solution the remaining jobs must be split evenly between $M_2$ and $M_3$ (Step 3). This means 
that such a schedule cannot be shorter than $\lceil(n+1)/4\rceil/s_2$ on $M_2$. Hence $C^*_{\max}(G) \geq \lceil(n+1)/4\rceil/s_2$. 
On the other hand, in the worst case \texttt{Algorithm 2} returns a schedule corresponding to an equitable coloring of $G$ (Step 5), which means that 
Alg$_2(G) \leq \lfloor(n+1)/3\rfloor/s_2$. Therefore
$$
\frac{\text{Alg}_2(G)}{C^*_{\max}(G)} \leq \frac{\lfloor(n+1)/3\rfloor/s_2}{C^*_{\max}(G)} \leq \frac{\lfloor(n+1)/3\rfloor/s_2}{\lceil(n+1)/4\rceil/s_2}
<\frac{(n+1)/3}{(n+1)/4} = \frac{4}{3}
$$

If $s_1 < 2s_2$ then the faster $M_1$ performs the bigger difference between the worst and best case is. In the worst case $s_1\cong 2s_2$. 
Then the length of optimal schedule is less than $(n+1)/s$. As previously, at worst our algorithm produces a schedule based on equitable 
coloring of $G$ whose length is at most $\lfloor(n+1)/3 \rfloor/s_2$. Hence Alg$_2(G)\leq \lfloor(n+1)/3\rfloor/s_2$ and
$$
\frac{\text{Alg}_2(G)}{C^*_{\max}(G)} \leq \frac{\lfloor(n+1)/3\rfloor/s_2}{C^*_{\max}(G)} \leq \frac{\lfloor(n+1)/3\rfloor/s_2}{(n+1)/s}\leq
\frac{s}{3s_2}<\frac{4s_2}{3s_2} = \frac{4}{3}
$$
and the thesis of the theorem follows. 
\end{proof}

\begin{theorem}
If  $3s_1/4 \leq s_2=s_3$ then \emph{\texttt{Algorithm 2}} almost always returns an optimal solution.
\end{theorem}
\begin{proof}
Frieze and Suen \cite{friez} showed that procedure \texttt{Greedy} finds an independent set of size $|I| \geq 0.432n - \epsilon n$ in almost all 
cubic graphs on $n$ vertices, where $\epsilon$ is any constant greater than 0. Notice that if it is really the case then $n^* < |I|$. 
Therefore \texttt{Algorithm 2} at first finds in Steps 2 and 3 a semi-equitable coloring of type $(|I|, \lceil(n-|I|)/2\rceil, \lfloor(n-|I|)/2\rfloor)$ 
and then transforms it into a semi-equitable coloring of type $(n^*, \lceil(n-n^*)/2\rceil, \lfloor(n-n^*)/2\rfloor)$ in Step 8. 
This completes the proof.   
\end{proof}

\section{Final remarks}

Can our results be generalized without changing the complexity status of the scheduling problem? The answer is $\ldots$ sometimes. Let us consider bicubic graphs for example. If arbitrary 
job lengths are allowed then the problem $Q3|G \in \mathcal{Q}_2|C_{\max}$
becomes NP-hard even if $s_1 = 2s_2 = 2s_3$. In fact, let $I_1, I_2$ be a decomposition of $G$ and suppose that the processing time 
$P_1(I_1) = 2P_2(I_2)$. Then all the jobs of $I_1$ should be assigned to $M_1$, which results in a schedule of 
length $P_1(I_2)$ on machine $M_1$. This schedule length equals $C^*_{\max}$ if and only if there is partition 
of the remaining jobs. Thus a solution to our scheduling problem solves an NP-complete 
PARTITION problem.

On the other hand, if all $n$ jobs are identical but $G$ is disconnected bicubic and $K_{3,3}$-free then 
\texttt{Algorithm 1} can be modified to obtain an optimal schedule in $O(n^2)$ time. First, we treat
$G = G_1 \cup G_2 \cup \ldots  \cup G_k$ as a connected graph and calculate the color sizes, say $n_1, n_2,$ and $n_3$
$(n_1+n_2+n_3 = n)$, that guarantee an optimal solution for $G$. Next, for each $i = 1,\ldots,k$ we split $G_i$
into independent sets $A_i, B_i$ and $C_i$, so that $\sum_{j=1}^i |B_j|/\sum_{j=1}^i |G_j|$ is as close to $n_2/n$ as possible,   
where $|G_j|$ is the order of subgraph $G_j$. The same should hold for sets $A_i$ and $C_i$ with $n_1/n$ and $n_3/n$, respectively. Similarly, we can extend 
\texttt{Algorithm 2} to deal with disconnected tricubic graphs in $O(n^3)$ time.

\end{document}